%% file: main.tex
\documentclass[11pt]{article}
\input{defs}

\usepackage{fullpage}

\DeclareMathOperator{\econg}{econg}

\DeclareMathOperator{\dist}{dist}

\title{A Simple Dynamic Spanner via APSP}

\newcommand*\samethanks[1][\value{footnote}]{\footnotemark[#1]}

\author{Rasmus Kyng\thanks{The research leading to these results has received funding from grant no. 200021 204787 of the Swiss National Science Foundation.} \\ ETH Zurich \\kyng@inf.ethz.ch \and Simon Meierhans\samethanks[1] \\ ETH Zurich \\ mesimon@inf.ethz.ch \and Gernot Zöcklein \\ ETH Zurich \\gzoecklein@ethz.ch}

\begin{document}

\maketitle

\begin{abstract}
    We give a simple algorithm for maintaining a $n^{o(1)}$-approximate spanner $H$ of a graph $G$ with $n$ vertices as $G$ receives edge updates by reduction to the dynamic All-Pairs Shortest Paths (APSP) problem. Given an initially empty graph $G$, our algorithm processes $m$ insertions and $n$ deletions in total time $m^{1 + o(1)}$ and maintains an initially empty spanner $H$ with total recourse $n^{1 + o(1)}$. When the number of insertions is much larger than the number of deletions, this notably yields recourse sub-linear in the total number of updates. 
    
    Our algorithm only has a single $O(\log n)$ factor overhead in runtime and approximation compared to the underlying APSP data structure. Therefore, future improvements for APSP will directly yield an improved dynamic spanner.

\end{abstract}

\section{Introduction}

Given a graph $G = (V, E)$, a $\delta$-approximate spanner $H$ is  a subgraph of $G$ with the same vertex set such that $\dist_H(u, v) \leq \delta \cdot \dist_G(u, v)$ for every pair of vertices $u, v \in V$. Althöfer, Das, Dobkin, Joseph and Soares developed the now-classical greedy spanner \cite{althofer1993spars}. 

Given an integer $\delta$, their algorithm incrementally constructs a $2\delta$-approximate spanner $H$ of $G = (V, E)$ as follows.
\begin{itemize}
    \item[1)] Initialize $H = (V, \emptyset)$.
    \item[2)] Consider the edges $(u, v) \in E$ in arbitrary order. If $\dist_H(u, v) > 2 \delta$, add $(u, v)$ to $H$. 
\end{itemize}
When this process finishes, every edge on the shortest path between two arbitrary vertices in $G$ can be replaced by a detour of length at most $2 \delta$ in $H$. The claimed approximation $\dist_H(u, v) \leq 2\delta \cdot \dist_G(u, v)$ directly follows. 

Crucially, the graph $H$ is also sparse. This follows from the folklore observation that high-girth graphs are sparse. 
The girth of a graph is the length of the shortest cycle. By construction, the girth of $H$ is $2\delta + 1$, which implies total edge count is bounded by $|V|^{1 + \frac{1}{\delta}}$.
In this article we consider the important setting $\delta \geq \log(|V|)$.

\begin{observation} \label{obs:girth}
    Every graph $H = (V, E_H)$ with girth $2\delta + 1$ contains at most $2|V|^{1 + \frac{1}{\delta}}$ edges.
\end{observation}
\begin{proof}
    Repeatedly remove vertices of degree at most $d =  |V|^{ \frac{1}{\delta}} + 1$ from $H$. This process removes at most $d \cdot |V|$ edges. Assume for the sake of a contradiction that there is some remaining graph $\tilde{H} \subseteq H$ after this process terminates. Fix an arbitrary remaining vertex $r$, and consider the breath-first search tree $T \subseteq \tilde{H}$ to depth $\delta$ from $r$. Since the girth of $\tilde{H}$ is at least $2\delta + 1$, no two vertices at depth $< \delta$ in $T$ share an edge in $\tilde{H}\setminus T$. Thus, the number of vertices in $T$ is at least
        $\sum_{i = 0}^{\delta} (d-1)^i > (d-1)^\delta = |V|$
    which is a contradiction.
\end{proof}

Despite the simplicity of this framework, previous deterministic \emph{dynamic algorithms} for maintaining a spanner $H \subseteq G$ as $G$ undergoes updates were based on expander decompositions \cite{maxflow, detmax}, which are arguably more natural for capturing cuts than distances. 

However, in \cite{bhattacharya_et_al:LIPIcs.ESA.2022.17} the authors observed that the greedy spanner construction can be used to maintain a spanner $H \subseteq G$ as $G$ receives edge deletions with low amortized recourse. As part of their almost-linear time incremental min-cost flow algorithm, \cite{CKL24} presented a computationally efficient algorithm for maintaining a spanner of a fully-dynamic graph $G$ based on the greedy spanner.\footnote{They crucially exploit that the recourse is sub-linear in the number of insertions. Expander decomposition based algorithms seem less suitable for achieving such a result.} Their algorithm is complicated by the fact that they need to process vertex splits in addition to edge updates. In this article, we present a simplified algorithm for maintaining a spanner of an edge-dynamic graph via an All-Pairs Shortest Paths (APSP) oracle.

\begin{theorem}  \label{thm:main}
    Given an initially empty edge-dynamic graph $G = (V, E)$ on $n = |V|$ vertices receiving $m$ edge insertions and up to $n$ edge deletions where $m \geq n$, there is an algorithm maintaining a  $n^{o(1)}$-approximate spanner $H$ with total update time $m^{1 + o(1)}$ and total recourse $n^{1 + o(1)}$. 
\end{theorem}

We note in particular that the total recourse is almost-optimal and only depends on the number of deletions and vertices, and is therefore completely independent of the number of insertions. We refer the reader to \Cref{def:apsp} and \Cref{thm:main_apsp} for a more fine grained analysis of the dependency on the underlying APSP data structure. 
Our algorithm can use any path-reporting approximate APSP data structure and only adds an $O(\log n)$ factor overhead in runtime and approximation.

\paragraph{Prior Work. } In \cite{rand_spanner}, the authors present a randomized dynamic spanner with poly-logarithmic stretch and amortized update time against an oblivious adversary.\footnote{Oblivious adversaries generate the sequence of updates at the start, i.e. independently from the random output of the algorithm. Adaptive adversaries can use the previous output to fool a randomized algorithm. } Then \cite{forster19} significantly reduced the poly-logarithmic overhead in both size and update time. In \cite{bernstein_et_al:LIPIcs.ICALP.2022.20}, the authors gave the first non-trivial algorithm for maintaining a dynamic spanner against an adaptive adversary. They maintain a spanner of near linear size with poly-logarithmic amortized update time and stretch. Deterministic algorithms for maintaining spanners of dynamic graphs with bounded degree that additionally undergo vertex splits were developed in the context of algorithms for minimum cost flow \cite{maxflow, detmax, CKL24}. These have sub-polynomial overhead.  

\section{Preliminaries}

\paragraph{Static and Dynamic Graphs.} We let $G = (V, E)$ denote unweighted and undirected graphs with vertex set $V$ and edge set $E$. When we want to refer to the vertex/edge set of some graph $G$, we sometimes use $V(G)$ and $E(G)$ respectively. 
For a graph $G = (V, E)$ and a pair of vertices $u,v \in V$ we let $\dist_G(u, v)$ denote the length of a shortest $uv$-path in $G$. 
We call a graph edge-dynamic if it receives arbitrary updates to $E$ (i.e. edge insertions and deletions). We refer to the total number of such edge updates a graph receives as the recourse. 

\paragraph{Edge Embeddings.} Given two graph $H$ and $G$ on the same vertex set such that $H \subseteq G$, we let $\Pi_{G \rightarrow H}$ be a function that maps edges in $E(G)$ to paths in $H$. Given an embedding $\Pi_{G \rightarrow H}$, we define the edge congestion $\econg(\Pi_{G \rightarrow H}, e) \defeq |\{f \in E(G): e \in \Pi_{G \rightarrow H}(f) \}|$ for $e \in E(H)$.

\section{Dynamic All-Pairs Shortest Paths}

We first define approximate \emph{All-Pairs Shortest Paths} (APSP) data structures for edge-dynamic graphs. 

\begin{definition}[APSP] \label{def:apsp}
    For an initially empty edge-dynamic graph $G = (V, E)$ a $\gamma$-approximate $\alpha$-update $\beta$-query APSP data structure supports the following operations.
    \begin{itemize}
        \item $\textsc{AddEdge}(u,v)$ / $\textsc{RemoveEdge}(u, v)$: Add/Remove edge $(u, v)$ from/to $G$ in (amortized) time $\alpha$.
        \item $\textsc{Dist}(u, v)$: Returns a distance estimate $\widetilde{\dist}(u,v)$ such that
        \begin{equation*}
            \dist(u, v) \leq \widetilde{\dist}(u,v) \leq \gamma \cdot \dist(u,v). 
        \end{equation*}
        in (amortized) time $\beta$.
        \item $\textsc{Path}(u,v)$: Returns a simple $uv$-path $P$ in $G$ of length $|P| \leq \textsc{Dist}(u,v)$ in time $\beta \cdot |P|$. %
    \end{itemize}
\end{definition}

We refer the reader to \cite{CZ23, kyng2023dynamic, haeupler2024dynamicdeterministicconstantapproximatedistance} for recent results on path-reporting dynamic APSP.
The parameters $\gamma, \alpha$ and $\beta$ should currently be thought of as subpolynomial factors $n^{o(1)}$ in the number of vertices $n$. 

\section{A Dynamic Greedy Spanner}

\paragraph{A Simple Low-Recourse Algorithm. }  We first present a very simple dynamic variant of the greedy spanner that merely limits the number of changes to $H$, i.e. its recourse. Then we describe the necessary changes for obtaining an efficient version using a dynamic APSP datastructure. 

Let $G = (V, E)$ be an initially empty graph on $n$ vertices. We maintain a $2 \log n$ spanner $H \subseteq G$ as follows. 
\begin{itemize}
    \item When an edge $(u,v)$ is inserted to $G$, we check if $\dist(u,v) > 2 \log n$. If so, we add it to $H$.
    \item When an edge $(u,v)$ is inserted to $G$, we check if $\dist(u,v) > 2 \log n$. If so, we add it to $H$. When an edge $(u,v)$ is deleted from $G$, we remove it from $H$ and re-insert all edges in $E(G)\setminus E(H)$ using the procedure described above.
\end{itemize}
This algorithm still maintains that the girth of $H$ is at least $2 \log n + 1$, and therefore the graph $H$ contains at most $O(n)$ edges at any point. Since at most $n$ edges leave $H$ due to deletions a total recourse bound of $O(n)$ follows directly. 

\paragraph{A Simple Efficient Algorithm. } To turn the above framework into an efficient algorithm, we maintain an embedding $\Pi_{G \rightarrow H}$ that maps each edge in $G$ to a short path in $H$ and thus certifies that such a path still exists. We then only have to re-insert all edges whose embedding paths use the deleted edge. To bound the number of re-insertions, we carefully manage the edge congestion of $\Pi_{G \rightarrow H}$. Finally, we use an APSP data structure (See \Cref{def:apsp}) to obtain distances and paths. In the following, we assume there are at most $m$ insertions, $n$ deletions and that $|V| = n$.

Our algorithm internally maintains two graphs: The current spanner $H$ and a not yet congested sub-graph $\hat{H} \subseteq H$. We maintain that the edge congestion is strictly less than $m/n$ for every edge in $\hat{H}$, and maintain access to distances and paths in $\hat{H}$ via a $\gamma$-approximate $\alpha$-update $\beta$-query APSP data structure $\mathcal{D}_{\hat{H}}$. 
\begin{itemize}
    \item When an edge $e$ is inserted, we check if $\mathcal{D}_{\hat{H}}$ returns distance at most $2\gamma \log n$ between the endpoints. 
    If this is the case, we query it for a witness path $P$ and store $P$ as the embedding path of $e$. Then we check if any of the edges on the path now have $m/n$ embedding paths using them, and if so we remove them from $\hat{H}$.
Otherwise, we add the edge $e$ to $H$ and $\hat{H}$.
    \item When an edge $e$ is deleted, we remove it from $\hat{H}$ and $H$ and re-insert all edges that now have a broken embedding path using the procedure above. 
\end{itemize}
See \Cref{algo:simple_spanner} for detailed pseudocode.

\begin{algorithm}[!ht]
\caption{$\textsc{DynamicSpanner}()$}
\label{algo:simple_spanner}
\SetKwProg{Globals}{global variables}{}{}
\SetKwProg{Proc}{procedure}{}{}
\Proc{$\textsc{Initialize}(n, m)$}{
    $G \gets (V, \emptyset)$; $H \gets (V, \emptyset)$; $\hat{H} \gets (V, \emptyset)$; $\Pi_{G \rightarrow H} \gets \emptyset$ \tcp*{Let $V = \{1, \ldots, n\}$.}
    Let $\mathcal{D}_{\hat{H}}$ be a $\gamma$-approximate $\alpha$-update $\beta$-query APSP datastructure (See \Cref{def:apsp})
}
\Proc{$\textsc{InsertEdge}(e = (u, v))$}{
    $G \gets G \cup (e)$ \\
    \If{$\mathcal{D}_{\hat{H}}.\textsc{Dist}(u,v) \leq \gamma \cdot 2 \cdot \log n$}{
        $P \gets \mathcal{D}_{\hat{H}}.\textsc{Path}(e)$; $\Pi_{G \rightarrow H}(e) \gets P$ \\
        \ForEach{$f \in P$}{
            \If{$\econg(\Pi_{G \rightarrow H}, f) \geq \frac{m}{n}$}{
                $\hat{H} \gets \hat{H} \setminus f$; $\mathcal{D}_{\hat{H}}.\textsc{RemoveEdge}(f)$
            }
        }
    }\Else{
        $H \gets H \cup (e)$; $\hat{H} \gets H \cup (e)$; $\Pi_{G \rightarrow H}(e) \gets (e)$; $\mathcal{D}_{\hat{H}}.\textsc{AddEdge}(e)$
    }
}
\Proc{$\textsc{DeleteEdge}(e = (u, v))$}{
    $G \gets G \setminus e$ \\
    \If{$e \in H$}{
        $F \gets \{f \in E(G): e \in \Pi_{G \rightarrow h}(f)\}$ \tcp*{Edges with broken embedding paths.}
        $G \gets G \setminus F$ \tcp*{Temporarily remove edges in $F$ from $G$}
        \ForEach{$f \in F$}{
            $\Pi_{G \rightarrow H}(f) \gets \emptyset$ \tcp*{Delete broken embeddings.}
        }
        \ForEach{$f \in F$}{
            $\textsc{InsertEdge}(f)$ \tcp*{Re-insert edges in $F$.}
        }
    }
}
\end{algorithm}

\paragraph{Analysis. } 

We first show that $H$ is a $2 \cdot \gamma \cdot \log n$-spanner. 

\begin{lemma}
    The graph $H$ maintained by $\textsc{DynamicSpanner}()$ (\Cref{algo:simple_spanner}) is a $2 \cdot \gamma \cdot \log n$-spanner throughout. 
\end{lemma}
\begin{proof}
    When an edge $e \in E(G)$ is inserted, it afterwards has a valid embedding path of length at most $2 \cdot \gamma \cdot  \log n$ by the description of $\textsc{InsertEdge}(e)$. Therefore the edge has a valid embedding path at this point. However, if this any edge on this path is deleted, the edge is re-inserted by the description of $\textsc{DeleteEdge}(e)$. This establishes the claim.
\end{proof}

We then prove that the total number of times the routine $\textsc{InsertEdge}(e)$ is called is bounded by $2m$. In the following, we assume without loss of generality that $m/n$ is an integer. 

\begin{claim} \label{clm:econg}
    For all $f \in H$, $\econg(\Pi_{G \rightarrow H}, f) \leq m/n$ throughout.
\end{claim}
\begin{proof}
    Whenever the edge congestion of an edge in $\hat{H}$ reaches $m/n$ it is removed from $\hat{H}$ and $\mathcal{D}_{\hat{H}}$. By the description of $\textsc{InsertEdge}()$ such an edge is never used in embedding paths again. Therefore the edge congestion of $\econg(\Pi_{G \rightarrow H}, f)$ is bounded by $m/n$ for every edge. 
\end{proof}

\begin{corollary}\label{corr:all_insertions}
    A sequence of at most $m$ external insertions and $n$ deletions causes at most $2m$ calls to $\textsc{InsertEdge}()$ in total.
\end{corollary}
\begin{proof}
    Directly follows from \Cref{clm:econg} and the description of $\textsc{DeleteEdge}()$.
\end{proof}

Given the previous corollary, we are ready bound the total recourse of $H$. 

\begin{claim}\label{clm:number_edges}
    $|E(H)| \leq O(\gamma \cdot n \log n)$ throughout. 
\end{claim}
\begin{proof}
    We bound the edges in $E(H) \setminus E(\widehat{H})$ and in $E(\widehat{H})$ separately, starting with the former. These edges were congested by $m/n$ paths when they were removed from $\hat{H}$. Since the total number of edges ever added is $2m$, the total cumulative congestion that all the edges ever experience is at most $4 \cdot m \cdot \gamma \cdot \log n$ by the description of $\textsc{InsertEdge}()$. Therefore, the total number of edges in $E(H) \setminus E(\widehat{H})$ is at most $4 \cdot n \cdot \gamma \cdot \log n$. Finally, by \Cref{def:apsp} the graph $\hat{H}$ has girth at least $2\log n + 1$, and therefore at most $O(n)$ edges by \Cref{obs:girth}. 
\end{proof}
\begin{corollary} \label{cor:recourse}
    The recourse of $H$ is $O(\gamma \cdot n \log n)$.
\end{corollary}
\begin{proof}
    Follows from \Cref{clm:number_edges} since at most $n$ edges get removed from $H$. 
\end{proof}

It remains to bound the total time spent processing updates. 

\begin{lemma}\label{lma:runtime}
    The algorithm processes $m$ insertions and $n$ deletions in total time $O(\beta \cdot \gamma \cdot m 
\cdot \log n + \alpha \cdot \gamma \cdot n \cdot \log n)$. 
\end{lemma}
\begin{proof}
    There are at most $2m$ calls to $\textsc{InsertEdge}()$. Each such call causes a distance and possibly a path reporting query to $\mathcal{D}_{\hat{H}}$. Whenever a path reporting query is called, the resulting path $P$ has length at most $2\gamma \cdot \log n$. This causes a total runtime of $O(\beta \cdot \gamma \cdot m \cdot \log n)$.

    Then, by \Cref{cor:recourse}, at most $O(\gamma \cdot n \log n)$ get added and therefore also removed from $\hat{H}$. These cause $O(\gamma \cdot n \log n)$ update calls to $\mathcal{D}_{\hat{H}}$. This causes total runtime $O(\alpha \cdot \gamma \cdot n \cdot \log n)$. 

    All other operations can be subsumed in the times above. 
\end{proof}

Our main theorem then follows directly.

\begin{theorem} \label{thm:main_apsp}
    Given a $\gamma$-approximate $\alpha$-update $\beta$-query APSP data structure, there is an algorithm that maintains a $2 \gamma \log n$-approximate spanner $H$ of an initially empty graph $G$ on $n$ vertices throughout a sequence up to $m$ insertions and $n$ deletions in total time $O(\beta \cdot \gamma \cdot m 
\cdot \log n  + \alpha \cdot \gamma \cdot n \cdot \log n)$. The total recourse of $H$ is $O(\alpha \cdot \gamma \cdot n \log n)$. 
\end{theorem}
\begin{proof}
    Follows from \Cref{cor:recourse} and \Cref{lma:runtime}.
\end{proof}

Finally, \Cref{thm:main} follows from \Cref{thm:main_apsp} in conjunction with any of the APSP data structures of \cite{CZ23, kyng2023dynamic, haeupler2024dynamicdeterministicconstantapproximatedistance}. 

\section*{Acknowledgments}

The authors are very grateful for feedback from Maximilian Probst Gutenberg and Pratyai Mazumder on a draft of this article that improved the presentation.  

\newpage
\bibliographystyle{alpha}
\bibliography{refs}
\end{document}

%% file: defs.tex
\usepackage[usenames,dvipsnames,svgnames,table]{xcolor}
\usepackage{enumitem}
\usepackage{graphicx}
\usepackage{fancybox}
\usepackage{comment}
\usepackage{xcolor}
\usepackage{nameref}
\usepackage{bbm}

\definecolor{ForestGreen}{rgb}{0.1333,0.5451,0.1333}
\definecolor{DarkRed}{rgb}{0.8,0,0}
\definecolor{Red}{rgb}{1,0,0}
\usepackage[linktocpage=true,
pagebackref=true,colorlinks,
linkcolor=ForestGreen,citecolor=ForestGreen,
bookmarks,bookmarksopen,bookmarksnumbered]
{hyperref}
\usepackage[nottoc,numbib]{tocbibind}

\usepackage{amsmath}
\usepackage{amssymb}
\usepackage{amsthm}

\usepackage[ruled, noend, linesnumbered]{algorithm2e}

\usepackage{subfig}

\usepackage{thm-restate}

\usepackage{float}
\usepackage{cleveref}

\newtheorem{theorem}{Theorem}[section]

\newtheorem{corollary}[theorem]{Corollary}
\newtheorem{lemma}[theorem]{Lemma}
\newtheorem{observation}[theorem]{Observation}

\newtheorem{claim}[theorem]{Claim}

\newtheorem{definition}[theorem]{Definition}

\newtheorem*{theorem*}{Theorem}
\newtheorem*{corollary*}{Corollary}
\newtheorem*{conjecture*}{Conjecture}
\newtheorem*{lemma*}{Lemma}
\newtheorem*{thm*}{Theorem}
\newtheorem*{prop*}{Proposition}
\newtheorem*{obs*}{Observation}
\newtheorem*{definition*}{Definition}

\newtheorem*{remark*}{Remark}
\newtheorem*{rec*}{Recommendation}

\newenvironment{fminipage}%
  {\begin{Sbox}\begin{minipage}}%
  {\end{minipage}\end{Sbox}\fbox{\TheSbox}}

\def\defeq{\stackrel{\mathrm{def}}{=}}

\renewcommand{\hat}{\widehat}
\renewcommand{\tilde}{\widetilde}

\DeclareFontFamily{U}{mathb}{\hyphenchar\font45}
\DeclareFontShape{U}{mathb}{m}{n}{<5> <6> <7> <8> <9> <10> gen * mathb
<10.95> mathb10 <12> <14.4> <17.28> <20.74> <24.88> mathb12}{}
\DeclareSymbolFont{mathb}{U}{mathb}{m}{n}
\DeclareMathSymbol{\rcirclearrow}{\mathbin}{mathb}{'367}

\newif\ifrandom
\randomtrue

\newcommand{\todolater}[1]{}

\DeclareUnicodeCharacter{2113}{$\ell$}